\documentclass[journal,10pt,twocolumn,twoside]{IEEEtran}
\usepackage{amsmath,amsfonts,amssymb,amsbsy,bm,paralist,theorem,color}
\usepackage{graphicx,algorithmic,algorithm,relsize}
\usepackage{multicol}
\usepackage{multirow}
\usepackage{subfigure}
\usepackage{caption}
\usepackage{cite}

\newtheorem{prop}{Proposition}

\newcommand{\h} {\mathbf{h}}
\newcommand{\Th} {\mathbf{\Theta}}
\newcommand{\G} {\mathbf{G}}

\newcommand{\w}{\mathbf{w}}
\newcommand{\vv}{\mathbf{v}}

\DeclareMathOperator*{\st}{s.t.}
\DeclareMathOperator*{\Tr}{Tr}
\graphicspath{{fig/}}

\definecolor{orange}{RGB}{255,107,0}
\definecolor{green}{RGB}{0,160,20}

\begin{document}
\title{Energy-Efficient Design of IRS-NOMA Networks}
%{Energy efficient resource allocation for NOMA downlink MEC networks with imperfect CSI}
\author{Fang Fang,~\IEEEmembership{Member,~IEEE}, Yanqing Xu,~\IEEEmembership{Member,~IEEE}, Quoc-Viet Pham,~\IEEEmembership{Member,~IEEE}, \\and Zhiguo Ding,~\IEEEmembership{Fellow,~IEEE}
\thanks{Fang Fang is with the Department of Engineering, Durham University, Durham DH1 3LE, U.K.(e-mail: fang.fang@durham.ac.uk).}
\thanks{Zhiguo Ding is with School of Electrical and Electronic Engineering, The University of Manchester, M13 9PL, UK (e-mail: fang.fang@manchester.ac.uk and zhiguo.ding@manchester.ac.uk).}
\thanks{Yanqing Xu is with the School of Science and Engineering, The Chinese University of Hong Kong, Shenzhen, Shenzhen 518172, China. (e-mail: xuyanqing91@gmail.com).}
\thanks{Quoc-Viet Pham is with Research Institute of Computer, Information and Communication, Pusan National University, Busan 46241, South Korea (e-mail: vietpq@pusan.ac.kr).}
}\maketitle

%##################################################################
\begin{abstract}
Combining intelligent reflecting surface (IRS) and non-orthogonal multiple access (NOMA) is an effective solution to enhance communication coverage and energy efficiency. In this paper, we focus on an IRS-assisted NOMA network and propose an energy-efficient algorithm to yield a good tradeoff between the sum-rate maximization and total power consumption minimization. We aim to maximize the system energy efficiency by jointly optimizing the transmit beamforming at the BS and the reflecting beamforming at the IRS. Specifically, the transmit beamforming and the phases of the low-cost passive elements on the IRS are alternatively optimized until the convergence. Simulation results demonstrate that the proposed algorithm in IRS-NOMA can yield superior performance compared with the conventional OMA-IRS and NOMA with a random phase IRS.
\end{abstract}
\vspace{-1em}
\section{Introduction}
Intelligent reflecting surface (IRS) has been envisioned as a revolutionary technology in the beyond fifth-generation (B5G) wireless network \cite{QWuMag2020}. Compared to the conventional wireless relaying technology, which regenerates and retransmits signals, IRS only reflects  signals as well as operating in a full-duplex mode with low energy consumption. By adjusting the phases of the low-cost passive reflecting elements integrated on the IRS, the reflected signal propagation can be collaboratively modified to improve the communication coverage, throughput and energy efficiency \cite{IRSMag2020,QWuTWC2019}. 

Non-orthogonal multiple access (NOMA) is considered as a key technology in B5G due to its high spectral efficiency \cite{KYangNOMAMag19} and high energy efficiency \cite{FangIEEETrans16}. Motivated by the advantages of IRS and NOMA, IRS has been proposed to combine with NOMA  \cite{ASIRSNOMAMag,BZhengIRSNOMA2020}. In this paper, we mainly focus on the downlink multiple-input-single-output (MISO) IRS-assisted NOMA network. Different from the sum-rate maximization \cite{GYangIRSNOMA2020,XMuIRSNOMA2019} and transmit power minimization \cite{JZhuIRSNOMA2020}, our goal is to achieve the optimal tradeoff between sum-rate maximization and power minimization in the downlink MISO IRS-NOMA network. We aim to maximize the amount of transmitted data bits per Hz with unit energy, which is measured by energy efficiency, an important performance metric for green communications \cite{JAnTWC2020,CHuangEEIRS}. By jointly optimizing the transmit beamforming at the base station (BS) and the reflected beamforming at the IRS, we propose an efficient algorithm to achieve the maximum energy efficiency of the IRS-NOMA network. %To our best knowledge, energy efficiency has not been well studied yet in IRS-NOMA systems. 
%The proposed algorithm can yield significant performance gain on energy efficiency, which is verified by simulation results. 

%IRS has recently received growing attention since it is compatible with various communication techniques, such as millimeter-wave \cite{VJMMwave2019} and non-orthogonal multiple access (NOMA) \cite{ZDingIRS2019,BZhengIRSNOMA2020}. NOMA is considered as an important technology to improve the spectral efficiency \cite{KYangNOMAMag19}.  One of the IRS applications is to provide the wireless transmission to the users in a dead zone, where direct transmission links from the base station (BS) to users are severely blocked by buildings, thick walls, etc. 

%Since there is no direct link between the BS and $U_2$. The main motivation of this work is that the BS can serve two users by using IRS NOMA with low interference due to the SIC technology at the receiver.  
\begin{figure}[t]
	\centering
	\graphicspath{{./figures/}}
	\includegraphics[width=0.8\linewidth]{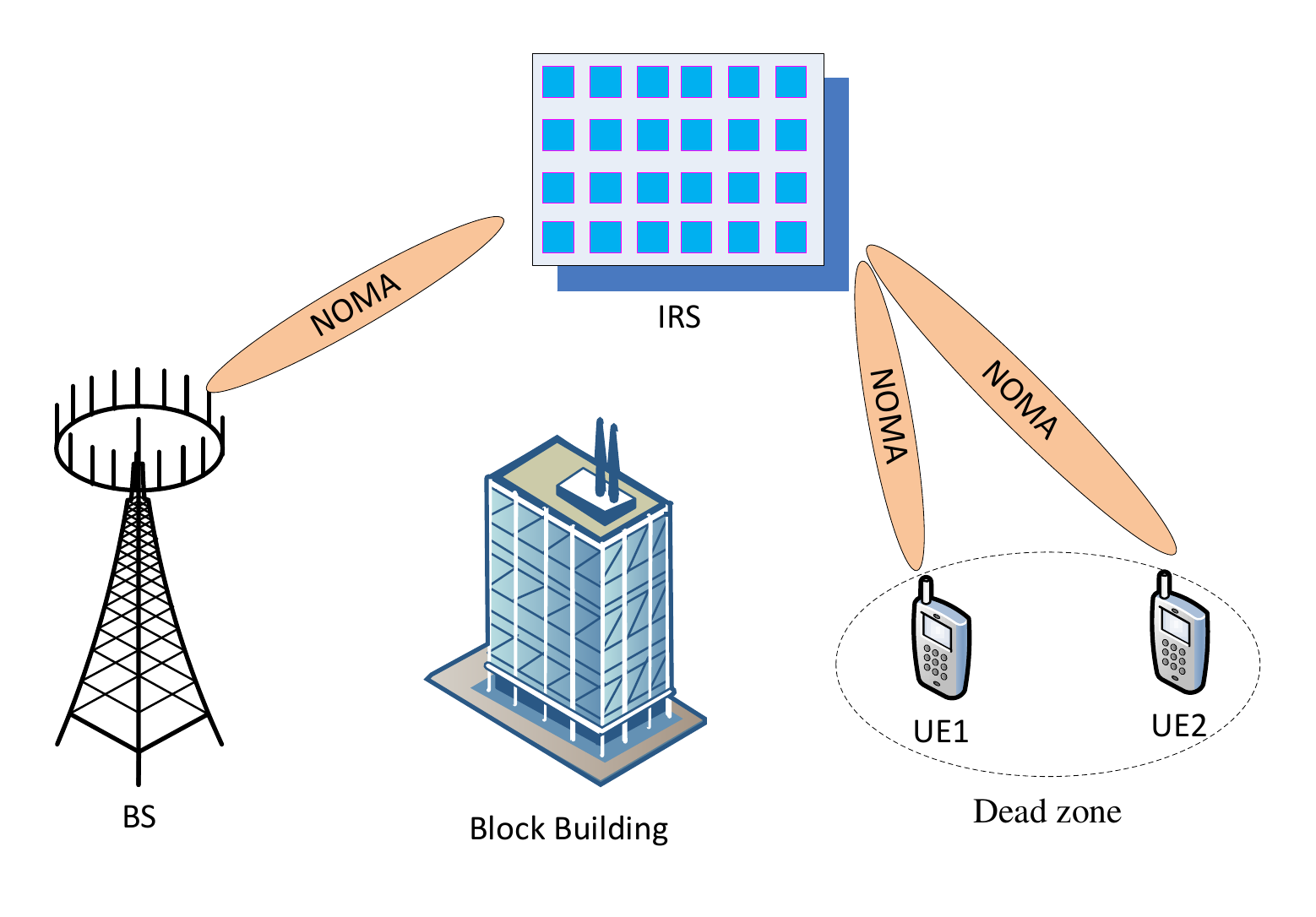}\\
	\caption{An IRS-NOMA MISO system.} \label{Fig0}
\end{figure}
\section{System Model and Problem Formulation}
\subsection{IRS Assisted NOMA}
As shown in Fig. 1, we consider a downlink MISO IRS-NOMA network, where an IRS composed of $N$ reflecting elements is implemented to assist the BS equipped with $M$ antennas to transmit signals to dead-zone single-antenna users. Due to the decoding complexity of successive interference cancellation (SIC) technology, we consider the number of users is two, $U_k, k=1,2$. In this system, the dead zone users cannot be served by the BS through direct links between the BS and users. We assume that the BS knows the perfect channel state information. Define $\Th=\text{diag}(\beta_1e^{j\theta_1},\cdots,\beta_Ne^{j\theta_N})$ as the reflection matrix for IRS where $\beta_{n}$ and $\theta_{n}$ respectively denote the amplitude reflection coefficient and the reflection phase shift. In this scenario, we adopt fixed amplitude reflection coefficients $\beta_n=1, \forall n$. The channel gains from the BS to the IRS and from the IRS to $U_k$ are denoted by $\G \in \mathbb{C}^{N\times M}$ and $\h_{k,r}^H \in \mathbb{C}^{1\times N}$, respectively, where $\h^H$ denotes the conjugate transpose of $\h$. Without loss of generality, the channel gain of the two users can be sorted as $|\h_{1,r}^H\Th\G|\geq|\h_{2,r}^H\Th\G|$. The BS broadcasts $\w_1s_1+\w_2s_2$, where $\w_m$ is the beamforming vector for $U_k, k\in\{1,2\}$, and $s_k$ is the information-bearing signal for $U_k$. Thus the received signals at $U_k$ can be respectively expressed as
\begin{equation}
	y_k=\h_{k,r}^H\Th\G(\w_1s_1+\w_2s_2)+n_k,
\end{equation}
where $n_k \sim {\mathcal{CN}}(0,\sigma^2)$ is the additive white Gaussian noise at $U_k$. Assume that the decoding order is ($U_2$, $U_1$). To ensure the performance of SIC, $U_1$ needs to successfully decode the signal of $U_2$. Hence the signal-interference-plus-noise ratio (SINR) to decode $U_2$'s message at $U_1$ is given by
\begin{equation}
	\Gamma_{2,1}=\frac{|\h_{1,r}^H\Th\G\w_2|^2}{|\h_{1,r}^H\Th\G\w_1|^2+\sigma^2}.
\end{equation} 
Define the SINR of $U_k$ as $\Gamma_k=\min\{\Gamma_{k,i}\}, \ \forall i\leq k$. The achievable rates of these two users can be written as 
\cite{ZChenTSP16MISONOMA}
\begin{subequations}
\begin{align}
	&R_1=\log_2(1+\Gamma_{1}),\\
	&R_2=\min\{\log_2(1+\Gamma_{2,2}),\ \log_2(1+\Gamma_{2,1})\} \label{R2},
\end{align}
\end{subequations}
where
\begin{equation}
\Gamma_{1}=\frac{|\h_{1,r}^H\Th\G \w_1|^2}{\sigma^2}\ \text{and}\ 	\Gamma_{2,2}=\frac{|\h_{2,r}^H\Th\G\w_2|^2}{|\h_{2,r}^H\Th\G\w_1|^2+\sigma^2}.
\end{equation}
\vspace{-0.8cm}
\subsection{Problem formulation}
 Let us define $\eta\in [0,1]$ as the power amplifier efficiency at the BS and denote the total circuit power at the BS by $P_c=MP_d+P_0$ and where $P_d$ is the dynamic power consumption and $P_0$ is the static power consumption. In the IRS-NOMA system, we aim to maximize the system energy efficiency, which is defined as a ratio of the system sum rate and the total power consumption. Considering the individual data rate constraints and the total transmit power budget, the energy efficiency maximization problem can be formulated as
\begin{subequations}\label{Prob:EE_1}
\begin{align}
\mathop {\max }\limits_{\Th,\w_1,\w_2} \quad & \frac{R}{\frac{1}{\eta}\sum\limits_{k=1}^2\|\w_k\|^2+P_c}\\
 \st\ \quad  & R_k\geq R_{k,\min}, k=1,2 \label{QoS_Cons},\\
& \sum\limits_{k=1}^2\|\w_k\|^2 \leq P_{\max}, \label{Pmax_Cons}\\
& 0\leq\theta_n\leq2\pi, n=1,\cdots,N, \label{theta_Cons}
\end{align}
\end{subequations}
where $R=R_1+R_2$, and $R_{k,\min}=\log_2(1+\Gamma_{k,\min})$ is the minimum data rate requirement for $U_k$, and where $\Gamma_{k,\min}=2^{R_{k,\min}}-1$ is the minimum SINR for $U_k$, which is a known parameter. Constraint \eqref{QoS_Cons} guarantees the QoS requirement for each user. Constraint \eqref{Pmax_Cons} limits the transmit power to $P_{\max}$. Constraint \eqref{theta_Cons} specifies the range of phase shift. However, $R_k$ is not jointly concave with respect to $\Th$ and $\w$. It is challenging to obtain the globally optimal solution to problem \eqref{Prob:EE_1} due to its non-convexity.
\vspace{-0.8em}
\section{Alternating Optimization Solution}
%\vspace{-0.6em}
 In this section, an alternating optimization-based algorithm is proposed to solve problem \eqref{Prob:EE_1} efficiently. Specifically, we decouple problem \eqref{Prob:EE_1} into beamforming optimization and phase shift optimization subproblems, and then solve the subproblems alternatively. Even though the alternating algorithm is widely used in the existing works \cite{QWuTWC2019,GYangIRSNOMA2020,JZhuIRSNOMA2020}, the proposed solutions to beamforming optimization and phase optimization in this work are different from the existing algorithms. In particular, the sequential convex approximation and successive convex approximation (SCA) are exploited to optimize $\w$, while a lower bound approximation and semi-definite relaxation (SDR) techniques are used to optimize $\Th$.  

\subsection{Beamforming Optimization}
For given phase shift $\Th$, problem \eqref{Prob:EE_1} is still non-convex. Inspired by sequential convex programming \cite{HMCuman2019}, we introduce a slack variable $t$ and equivalently transform problem \eqref{Prob:EE_1} as
\begin{subequations}\label{Prob:EE_t}
\begin{align}
\mathop {\max }\limits_{\w_1,\w_2,t} \quad &t,\\
 \st\ \quad  &\frac{R}{\frac{1}{\eta}\sum\limits_{k=1}^2\|\w_k\|^2+P_c}\geq t, \label{ConEE_t}\\
 &\eqref{QoS_Cons},\eqref{Pmax_Cons}.
\end{align}
\end{subequations}
To deal with the non-convex set \eqref{ConEE_t}, we introduce another slack variable $\rho$, then constraint \eqref{ConEE_t} can be equivalently replaced by
\begin{subequations}\label{ConEE_t_eq}
\begin{align}
	&R\geq t\rho, \label{Con_tp}\\
	&\frac{1}{\eta}\sum\limits_{k=1}^2\|\w_k\|^2+P_c\leq \rho. \label{pp}
\end{align}
\end{subequations}
Note that the equivalence between \eqref{ConEE_t} and \eqref{ConEE_t_eq} can be guaranteed since \eqref{ConEE_t_eq} must hold with equality at the optimum. To further track the convexity of constraint \eqref{Con_tp}, a set of new slack variables $\bm{\gamma}=[\gamma_1,\gamma_2]^T$ is introduced. Then constraint \eqref{Con_tp} can be expressed as
 \begin{subequations}\label{8}
\begin{align}
	&\sum\limits_{k=1}^2\log_2(\gamma_k)\geq t\rho, \label{Cons_Rk}\\
	&1+\Gamma_k\geq \gamma_k, k=1,2.\label{Cons_r}
\end{align}
\end{subequations}
By introducing another variable $\bm{\delta}=[\delta_1,\delta_2]^T$, constraint \eqref{Cons_Rk} can be equivalently represented by
\begin{subequations}\label{9}
\begin{align}
	&\sum\limits_{k=1}^2\delta_k\geq t\rho, k=1,2, \label{Cons_tp_2} \\
	&\gamma_k\geq 2^{\delta_k}, k=1,2.\label{con_gamma} \end{align}
\end{subequations}
Based on \eqref{8} and \eqref{9}, constraint \eqref{Con_tp} can be equivalently transformed to 
\begin{eqnarray}
\begin{aligned}\label{7a_2}
	\eqref{Con_tp}\Leftrightarrow\left\{\begin{array}{lr}
           \text{(8b):}\   1+\Gamma_k\geq \gamma_k,   \\    \text{(9b):}\          \gamma_k\geq 2^{\delta_k},\\
             \text{(9a):}\ \sum\limits_{k=1}^2\delta_k\geq t\rho. 
           \end{array}
             \right.
\end{aligned}
\end{eqnarray}
It is obvious that constraint \eqref{pp} is a convex set since it can be written as a second-order cone (SOC) representation:
\begin{equation}
\begin{aligned}
   % &\frac{\rho+1}{2}\geq \left\|\left[\frac{\rho+1}{2},\hat {\rho}\right]\right\|_2\\
	\frac{\eta\rho-\eta P_c+1}{2}\geq \left\|\left[\frac{\eta\rho-\eta P_c-1}{2},\w_1^T,\w_{2}^T\right]^T        
\right\|_2.
\end{aligned}
\end{equation}
 By introducing another variable $\bm{\beta}=\{\beta_{i,k}\},(i\leq k, k=1,2)$, constraint \eqref{Cons_r} can be relaxed to
\begin{subequations}\label{Con_hw}
	\begin{align}
	&|\h_{i}^H \w_k|^2\geq (\gamma_k-1)\beta_{i,k},\  i\leq k, k=1,2,\label{12a}\\ 
	& |\h_{i}^H \w_{k-1}|^2+\sigma^2\leq \beta_{i,k},\  i\leq k, k=1,2,\label{12b} \\ \nonumber
	\end{align}
\end{subequations}
 %\vspace{-0.1cm}
where $\h_1^H=\h_{1,r}^H\Th\G$ and $\h_2^H=\h_{2,r}^H\Th\G$. An arbitrary phase rotation of the beamforming vectors can be added to make the imaginary part of $\h_i^Hw_k$ to be zero, which does not affect the value of SINR. Thus $\h_i^Hw_k$ can be chosen to be real. The inequalities \eqref{12a} can be rewritten as
\begin{equation}
		\Re(\h_i^H\w_k)\geq \sqrt{(\gamma_k-1)\beta_{i,k}},\ \Im(\h_{i}^H\w_k)=0,\  i\leq k, k=1,2 \label{gammabeta}
\end{equation}
and the inequalities \eqref{12b} can be rewritten as
\begin{equation}
		\beta_{i,k}\geq |\h_{i}^H \w_{k-1}|^2+\sigma^2,\ i\leq k, k=1,2 \label{beta}
\end{equation}
where $|\h_{i}^H \w_{i-1}|^2=0$ when $i=1$. Similar to constraint \eqref{pp}, constraint \eqref{QoS_Cons} can be reformulated as SOC:
\begin{equation}
	\frac{\Re(\h_{i}^H\w_k)}{\Gamma_k^{\min}}\geq \left\|\begin{array}{ccr}
\h_{i}^H[\w_1,\cdots,\w_{k-1}]        \\
   \sigma^2
\end{array}\right\|_2, \Im(\h_{i}^H\w_k)=0.\label{15}
\end{equation}
An arbitrary phase rotation can be added to $\w$, which will not affect the optimality of the solution. Therefore, problem \eqref{Prob:EE_t} can be equivalently transformed to 
\begin{subequations}\label{Prob:EE_t3}
\begin{align}
\mathop {\max }\limits_{t,\rho,\w,\bm{\gamma},\bm{\delta},\bm{\beta}}  &\quad t\\
 \st\quad  &\Re(\h_i^H\w_k)\geq \sqrt{(\gamma_k-1)\beta_{i,k}}, i\leq k, k=1,2,
\label{Nonconvex1}\\
 &\sum\limits_{k=1}^2\delta_k\geq t\rho, k=1,2,  \label{Nonconvex2}\\
 &\eqref{15},\eqref{Pmax_Cons},\eqref{pp},\eqref{con_gamma},\eqref{beta}.\label{Convex_Cons1}\\\nonumber
\end{align}
\end{subequations}
Next, we analyze the convexity of constraints. We note that the constraints in \eqref{Convex_Cons1} are convex while constraints \eqref{Nonconvex1} and \eqref{Nonconvex2} are non-convex. In the following, we propose to use SCA to transform the non-convex constraints to convex approximation expressions. Performing the first-order Taylor approximation, constraint \eqref{Nonconvex1} can be written as
\begin{equation}
	\begin{aligned}
		\Re(\h_i^H\w_k)\geq & \sqrt{(\gamma_k^{(l)}-1)\beta_{i,k}^{(l)}}+\frac{1}{2}\sqrt{\frac{\beta_{i,k}^{(l)}}{\gamma_k^{(l)}-1}}(\gamma_k-\gamma_k^{(l)})\\
		&+\frac{1}{2}\sqrt{\frac{\gamma_k^{(l)}-1}{\beta_{i,k}^{(l)}}}(\beta_{i,k}-\beta_{i,k}^{(l)}),\label{Nonconvex1_Taylor}
	\end{aligned}
\end{equation}
where $\gamma_k^{(l)}$ and $\beta_{i,k}^{(l)}$ are the value of the variable $\gamma_k$ and $\beta_{i,k}$ after the $l$-th iteration in the proposed SCA-based algorithm. Similarly, we exploit the first-order Taylor approximation to relax constraint \eqref{Nonconvex2} as
\begin{equation}
	\sum\limits_{k=1}^2\delta_k \geq t^{(l)}\rho^{(l)}+\rho^{(l)}(t-t^{(l)})+t^{(l)}(\rho-\rho^{(l)}),\label{Nonconvex2_Taylor}
\end{equation}
where $t^{(l)}\ \text{and}\ \rho^{(l)}$ denote the values of $t$ and $\rho$ after the $l$-th iteration. Therefore, given the optimized values from the $l$-th iteration, the original problem \eqref{Prob:EE_1} can be approximately transformed at $(l+1)$-th iteration to the following problem:
\begin{subequations}\label{Prob:EE_t4_Taylor}
\begin{align}
\mathop {\max }\limits_{t,\rho,\w,\bm{\gamma},\bm{\delta},\bm{\beta}}  &\quad t\\
 \st\quad  &\Re(\h_i^H\w_k)\geq \textit{FTA}^{(l)}\left(\gamma_k,\beta_{i,k}\right), i\leq k, k=1,2,
\label{Nonconvex14}\\
 &\sum\limits_{k=1}^2\delta_k\geq \textit{FTA}^{(l)}\left(t\rho\right),  \label{Nonconvex24}\\
 &\eqref{QoS_Cons},\eqref{Pmax_Cons},\eqref{pp},\eqref{con_gamma},\eqref{beta}\label{Convex_Cons}\\\nonumber
\end{align}
\end{subequations}
%\vspace{-0.1cm}
where $\textit{FTA}^{(l)}(\cdot)$ stands for the first-order Taylor approximation of the variable after the $l$-th iteration, i.e., the right hand side of inequalities of \eqref{Nonconvex1_Taylor} and \eqref{Nonconvex2_Taylor}. Thus we propose an SCA-based algorithm to solve the beamforming optimization subproblem (shown as steps 3-8 in Algorithm 1). Note that the Dinkelbach method can be also used to solve the fractional form beamforming optimization subproblem, but it turns out that the SCA-based method can achieve better performance than the Dinkelbach method \cite{HMCuman2019}.

%\begin{subequations}
%	\begin{align}
%		\mathop {\max }\limits_{\Delta} \quad &t\\
% \st \ \quad  &\sum\limits_{k=1}^2\|\w_k\|^2+P_c\leq \rho \text{(7b)}\\
% &\gamma_k\geq 2^{\delta_k}, k=1,2\text{(9b)} \\
%	&	\sum\limits_{k=1}^2\delta_k \geq t^{(l)}\rho^{(l)}+\rho^{(l)}(t-t^{(l)})+t^{(l)}(\rho-\rho^{(l)})\text{(10)}\\
%	\mathbb{R}(\h_i^H\w_k)\geq & \sqrt{(\gamma_k-1)\beta_{i,k}}+\frac{1}{2}\sqrt{\frac{\beta_{i,k}^{(l)}}{\gamma_k^{(l)}-1}}(\gamma_k-\gamma_k^{(l)})\beta_{i,k}\\\nonumber
%		&+\frac{1}{2}\sqrt{\frac{\gamma_k^{(l)}-1}{\beta_{i,k}^{(l)}}}(\beta_{i,k}-\beta_{i,k}^{(l)}).\text{(13)}\\ \nonumber
%		&\left\|\begin{array}{ccr}
%\h_{1,d}^H\w_1        \\
%   \frac{\psi_{1,2}-1}{2}
%\end{array}\right\|_2\leq \frac{\psi_{1,2}+1}{2},\ 	\left\|\begin{array}{ccr}
%\h_2^H\w_1        \\
%   \frac{\psi_{2,2}-1}{2}
%\end{array}\right\|_2\leq \frac{\psi_{2,2}+1}{2}. \text{(14)}\\
%&\frac{1}{\Gamma_k^{\min}}\mathbb{R}(\h_{i,h}^H\w_k)\geq \left\|\begin{array}{ccr}
%\h_{i}^H[\w_1,\cdots,\w_{k-1}]        \\
%   \sigma^2
%\end{array}\right\|_2. \text{(15)}
%	\end{align}
%\end{subequations}
The original transmit beamforming optimization problem \eqref{Prob:EE_t} can be solved by iteratively solving the above problem \eqref{Prob:EE_t4_Taylor}, which is convex. Specifically, we first initialize the optimized variables. To guarantee the feasibility and convergence of the above problem, we chose the initial optimized variables by evaluating the beamforming vectors and satisfying all the constraints. It is important to find the initial values of $\left\{t^{(0)},\rho^{(0)},\w^{(0)},\bm{\gamma}^{(0)},\bm{\beta}^{(0)}\right\}$ since the convergence of the proposed SCA-based method is sensitive to the initial points. To do so, we can solve a simple feasibility problem: $\rm{Find}\{\w|\eqref{QoS_Cons},\eqref{Pmax_Cons}\}$ and denote the obtained solution by $\w^{(0)}$. Then $\gamma_k^{(0)}$ and $\beta^{(0)}$ can be computed by replacing the inequalities of \eqref{gammabeta} and \eqref{beta} with equalities. At last, initial $\rho^{(0)}$ and $t^{(0)}$ can be calculated by $\rho^{(0)}=\frac{1}{\eta}\sum\nolimits_{k=1}^2\|\w_k^{(0)}\|^2+P_c$ and $t^{(0)}=\sum\nolimits_{k=1}^2\gamma_k^{(0)}$. At each iteration, we solve problem \eqref{Prob:EE_t4_Taylor} given by the values of the optimized variables by the last iteration. The progress will terminate until its convergence with the tolerance $\epsilon=0.001$.

\subsection{Phase shift optimization}
In this section, we focus on the phase optimization. The phase optimization problem is generally transformed into a feasibility problem in the existing works \cite{QWuTWC2019,GYangIRSNOMA2020,JZhuIRSNOMA2020}. However, in this work, given by the transmit beamforming vectors $\{\w_k\}$ obtained from the beamforming optimization, the energy efficiency maximization problem \eqref{Prob:EE_1} is reduced to a sum-rate maximization problem:
 \begin{subequations}\label{SubProb:theta}
\begin{align}
\mathop {\max }\limits_{\Th} \quad &  R_1+R_2\label{obj}\\
 \st\ \quad  & R_k\geq R_{k,\min}, k=1,2, \\
& 0\leq\theta_n\leq2\pi, n=1,\cdots,N.
\end{align}
\end{subequations}
Note that \eqref{obj} is a $\max \{\min\{\cdot\}\}$ function due to \eqref{R2}. To address this problem, we first transform it to a tractable one. Let $\h_{k,r}^H\Th\G\w_j=\vv^H\mathbf{a}_{k,j}$ where $\vv=[v_1,\cdots,v_N]^H=[e^{j\theta_1},\cdots,e^{j\theta_N}]^H$ and $\mathbf{a}_{k,j}={\rm{diag}}(\h_{k,r}^H)\G\w_{j}$. 
By introducing slack variable $z_1$, problem \eqref{SubProb:theta} can be equivalently transformed to           
\begin{subequations}\label{Prob:t1}
	\begin{align}
		\mathop {\max }\limits_{\vv,z_1} \quad & z_1\\ 
\st \quad  & R_1+R_2^1\geq z_1,\\
&R_1+R_2^2\geq z_1, \label{nonconvex}\\
 &  \frac{|\vv^H\mathbf{a}_{1,1}|^2}{\sigma^2} \geq \Gamma_{1,\min},\\ 
 &\frac{|\vv^H\mathbf{a}_{k,2}|^2}{|\vv^H\mathbf{a}_{k,1}|^2+\sigma^2} \geq \Gamma_{2,\min}, \forall k,=1,2,\\
& |\vv_n|=1, 0\leq\theta_n\leq2\pi, n=1,\cdots,N.\label{21e}
	\end{align}
\end{subequations}
However, problem \eqref{Prob:t1} is non-convex due to the non-convex constraint \eqref{nonconvex}. To make this problem tractable, we propose Proposition 1 to approximate constraint \eqref{nonconvex} to a convex set by its lower bound. 
\begin{prop}

 Problem \eqref{Prob:t1} can be solved by the following transformed problem \eqref{Prob:t2}. The optimal solution to the transformed problem \eqref{Prob:t2} is guaranteed to be a feasible solution to the original phase optimization problem \eqref{Prob:t1}.
 \begin{subequations}\label{Prob:t2}
	\begin{align}
		\mathop {\max }\limits_{\vv,z_2} \quad & z_2\\ \label{21b}
%\st \quad  & \Tr(\mathbf{H}_1\mathbf{W}_1)+\Tr(\mathbf{H}_1\mathbf{W}_2)\geq z_2\sigma^2,\\
%&\Tr(\mathbf{H}_2\mathbf{W}_1)+\Tr(\mathbf{H}_2\mathbf{W}_2)\geq z_2\sigma^2,\\
\st \quad &\frac{|\vv^H\mathbf{a}_{1,1}|^2+|\vv^H\mathbf{a}_{1,2}|^2}{\sigma^2}\geq z_2, \\\label{22a}
& \frac{|\vv^H\mathbf{a}_{2,1}|^2+|\vv^H\mathbf{a}_{2,2}|^2}{\sigma^2}\geq z_2,\\
 &  \frac{|\vv^H\mathbf{a}_{1,1}|^2}{\sigma^2} \geq \Gamma_{1,\min},\\ 
 &\frac{|\vv^H\mathbf{a}_{k,2}|^2}{|\vv^H\mathbf{a}_{k,1}|^2+\sigma^2} \geq \Gamma_{2,\min},  k= 1,2,\\ 
& |\vv_n|=1, 0\leq\theta_n\leq2\pi, n=1,\cdots,N.\label{21e}
	\end{align}
\end{subequations} 
\end{prop}
\begin{proof}
Since $\Tr(\mathbf{H}_1\mathbf{W}_2)\geq \Tr(\mathbf{H}_2\mathbf{W}_2)$, we have $R_1+R_2^2	\geq  (R_1+R_2^2)_{low}$ where $(R_1+R_2^2)_{low}=\log_2(1+(\Tr(\mathbf{H}_2\mathbf{W}_1)+\Tr(\mathbf{H}_2\mathbf{W}_2))/\sigma^2)$ is the lower bound of $R_1+R_2^2$.	 Thus the solution satisfying  $R_1+R_2^2\geq(R_1+R_2^2)_{low}$ and $(R_1+R_2^2)_{low}\geq z_1 $ must satisfy $R_1+R_2^2\geq z_1$. Let $z_2=2^{z_1}-1$, problem \eqref{Prob:t1} can be transformed to  \eqref{Prob:t2}.
\end{proof}  
Note that problem \eqref{Prob:t2} is non-convex due to the non-convexity of constraint \eqref{21e}. Constraints in \eqref{Prob:t2} can be transformed into quadratic form. Let $\mathbf{A}_{k,j}=\mathbf{a}_{k,j}\mathbf{a}_{k,j}^H, k=1,2$. We have $|\vv^H\mathbf{a}_{1,1}|^2=\Tr(\mathbf{V}A_{1,1})$ where $\mathbf{V}=\vv\vv^H$  satisifying $\mathbf{V}\succeq 0$ and $\rm{Rank}(\mathbf{V})=1$. Due to the non-convexity of Rank-one constraint, SDR can be applied to relax problem \eqref{Prob:t2} to 
\begin{subequations}\label{Prob:Trace}
	\begin{align}
		\mathop {\max }\limits_{\mathbf{V}, z_2} \quad & z_2\\
\st \quad 
 & \Tr(\mathbf{V}\mathbf{A}_{k,1})+\Tr(\mathbf{V}\mathbf{A}_{k,2})\geq z_2\sigma^2,\ k= 1,2, \label{23b}\\
 &\Tr(\mathbf{V}\mathbf{A}_{1,1}) \geq \Gamma_{1,\min}\sigma^2, \label{23c}\\
  &\Tr(\mathbf{V}\mathbf{A}_{k,2}) \geq \Gamma_{2,\min}(\Tr(\mathbf{V}\mathbf{A}_{k,1})+\sigma^2),\ k= 1,2,\label{23d}\\
 % &\Tr(\mathbf{V}\mathbf{A}_{2,2}) \geq \Gamma_{2,\min}(\Tr(\mathbf{V}\mathbf{A}_{2,1})+\sigma^2) \label{24f}\\
 &\rm{diag}(\mathbf{V})=\mathbf{1}_{N\times1},\label{23e}\\
%&\mathbf{V}_{n,n}=1, n=1,\cdots,N.\label{24g}\\
&\mathbf{V}\succeq 0,
	\end{align}
\end{subequations}
where $\rm{diag}(\mathbf{V})$ returns a column vector of the main diagonal elements of $\mathbf{V}$, and $\mathbf{1}_{N\times1}$ is a vector with size $N\times 1$ and ones on all elements. Problem \eqref{Prob:Trace} is a semi-definite programming (SDP) problem. Thus it can be solved optimally by the existing solvers such as CVX in Matlab. %However, the optimal solution to problem \eqref{Prob:Trace} is generally not a Rank-one solution. Therefore, the Gaussian randomization method can be used to find a feasible solution to problem \eqref{Prob:t2} \cite{AMSoGuassianRandom}. 

 According to Theorem 3.2 in \cite{YHuangTSP2010Rankone}, we can conclude that the optimal solution $\mathbf{V}^*$ to problem \eqref{Prob:Trace}, must satisfy the following constraint:
\begin{equation}
	{\rm{Rank}^2}(\mathbf{V}^*)+{\rm{Rank}^2}(z_2^*)\leq 6.
\end{equation}
First, we note that ${\rm{Rank}^2}(z_2^*)\leq 1$. When ${\rm{Rank}}^2(z_2^*)=0$, i.e., $t^*=0$, we have $\mathbf{V}^*=0$ which violates constraint \eqref{23b}. Thus, we have ${\rm{Rank}}^2(z_2^*)= 1$ and $\rm{Rank}(\mathbf{V}^*)=1\  \text{or}\  2$. When $\rm{Rank}(\mathbf{V}^*)=1$, $\vv^*$ can be found by singular value decomposition (SVD). Otherwise the Gaussian randomization method can be applied to generate a high quality approximate Rank-one solution. Interestingly, from our simulation results, problem \eqref{Prob:Trace} can always yield a Rank-one solution. This is because that condition \eqref{23b} when $k=1$ can be always included in the case $k=2$ due to $(\mathbf{A}_{1,1}+\mathbf{A}_{1,2})-(\mathbf{A}_{2,1}+\mathbf{A}_{2,2})\succeq 0$. Similar to condition \eqref{23d}, thus we have ${\rm{Rank}^2}(\mathbf{V}^*)+{\rm{Rank}}^2(z_2^*)\leq 4$, and  $\rm{Rank}(\mathbf{V}^*)=1$.
%Since $\rm{Rank}^2(z_2^*)=1$ and $\rm{Rank}(\mathbf{V}^*)$ must be integer, we have $\rm{Rank}(\mathbf{V}^*)=1$. One $\mathbf{V}$ is obtained, $\vv^*$ can be found by singular value decomposition (SVD). It is worth noting that there are multiple $\vv^*$ for one optimal $\mathbf{V}^*$. 
\begin{algorithm}[tp] \small
	\caption{Alternating SCA and SDR based  Algorithm }\label{Alg1}
	\begin{algorithmic}[1] 
		\STATE {{\bf Initialization:} Set $\Th={\Th}_0$ the outer iteration number $l=1$.}%, the accuracy $\epsilon=10^{-4}$.}
		\STATE {\bf repeat}\\
		\STATE {\bf Beamforming Optimization: }{Initialize the inner iteration number $i=1$ and $\left\{t^{(0)},\rho^{(0)},\w^{(0)},\bm{\gamma}^{(0)},\bm{\beta}^{(0)}\right\}$.}
		\STATE {Given $\Th=\rm{diag}(\vv)$}
		\WHILE {$t^{(i)}-t^{(i-1)}> \epsilon$}
		\STATE {Obtain the optimal solution $\left\{t^{(i)},\rho^{(i)},\w^{(i)},\bm{\gamma}^{(i)},\bm{\beta}^{(i)}\right\}$ to problem  \eqref{Prob:EE_t4_Taylor} with $\left\{t^{(i-1)},\rho^{(i-1)},\w^{(i-1)},\bm{\gamma}^{(i-1)},\bm{\beta}^{(i-1)}\right\}$.}
		\ENDWHILE
        \STATE {{\bf output:} $t^{(l)}=t^{(i)}$ and $\w^{(l)}=\w^{(i)}$.}\\
       \STATE  {\bf Phase Shift Optimization: }
		{Given by $\w^{(l)}$, obtain $\mathbf{V}^*$ by solving problem \eqref{Prob:Trace}}. 
		%\STATE {\bf Output:} $\mathbf{V}^*=\mathbf{V}$
		\STATE {Find $\vv^*$ by SVD of $\mathbf{V}^*$ or the Gaussian randomization method.}
		\STATE Calculate {$EE^{(l)}$} by $\vv^*$ and $w^{(l)}$ and let $l=l+1$.
	%	\ENDWHILE
	\STATE {\bf until} $EE^{(l)}-EE^{(l-1)}<\epsilon$ or problem \eqref{Prob:Trace} is infeasible.
     
		%\STATE {{\bf Output:} $\alpha_T^*=\frac{\alpha_T^{\min}+\alpha_T^{\max}}{2}$, $\bm{\beta}^*$ and $\bm{p}^*$.}
	\end{algorithmic}\label{Al1}
\end{algorithm}

\subsection{Algorithm Design}
The details of the proposed alternating optimization algorithm are shown in Algorithm \ref{Alg1}, in which problem \eqref{Prob:EE_t4_Taylor} and problem \eqref{Prob:Trace} are alternatively solved until the convergence metric is triggered. Denote the value of energy efficiency in problem (19) based on a solution ($\w,\Th$) by $EE(\w,\Th)$. If there exists a solution to problem (23), i.e., ($\w^{(l)},\Th^{(l+1)}$), then ($\w^{(l)},\Th^{(l)}$) and ($\w^{(l+1)},\Th^{(l+1)}$) are both feasible solutions to problem (19) in the $(l)$-th iteration and the $(l+1)$-th iteration, respectively. Therefore, we have $EE(\w^{(l+1)},\Th^{(l+1)})\geq EE(\w^{(l)},\Th^{(l+1)})$ and $EE(\w^{(l)},\Th^{(l+1)}) \geq EE(\w^{(l)},\Th^{(l)})$. The first inequation holds because $\w^{(l+1)}$ is the optimized solution to problem (19) for given $\Th^{(l+1)}$. The second inequation holds because the phase shift optimization problem is to maximize the system sum rate, shown as the objective function of problem (23). This increases the numerator of EE. Therefore, Algorithm \ref{Al1} is guaranteed to converge.

%At each iteration, the initial value should be the optimized value from the last iteration. Specifically, given by the initialized $\Th_0$, the algorithm starts to solve beamforming optimization subproblem by iteratively solving problem \eqref{Prob:EE_t4_Taylor} (shown as steps 5-8 in Algorithm \ref{Al1}). The convergence of beamforming optimization can be achieved within the inner iterations. Based on the obtained $\w^*$, the corresponding phase shift $\mathbf{V}^*$ can be obtained by solving problem \eqref{Prob:Trace}.
\section{Simulation Results}
In this section, numerical results are provided to demonstrate the superior performance of the proposed algorithm compared with IRS-OMA and IRS-NOMA with rand phase schemes. Note that the benchmark IRS-OMA scheme is also optimized by the SCA-based method and the SDR method. The channels involved in this work are Rician fading channels \cite{QWuTWC2019}. The line-of-sight and Rayleigh fading components are set to $G^{los}=G^{nlos}=2$. The path loss exponents of the BS-IRS, IRS-user are respectively set to be 2.2 and 2.5. The distance between the BS and IRS is $d_{BI}=40$ m and the distances between IRS to $U_1$ and $U_2$ are $d_{IU1}=10$ m and $d_{IU2}=20$ m. Without loss of generality, we assume that all users have the same SINR requirement $\Gamma_{k,\min}=10$ dB, $\forall k$. Other parameters are set as $\sigma^2=-80$ dBm and $\eta=0.6$.
\begin{figure}[t]
	\centering	
	\graphicspath{{./figures/}}\includegraphics[width=0.84\linewidth]{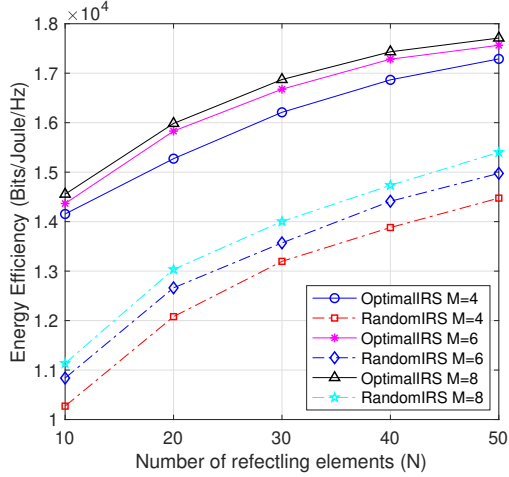}\\
	\caption{EE versus the number of reflecting elements $N$.} \label{Fig1}
\end{figure}
\begin{figure}[t]
	\centering
	\graphicspath{{./figures/}}
	\includegraphics[width=0.84\linewidth]{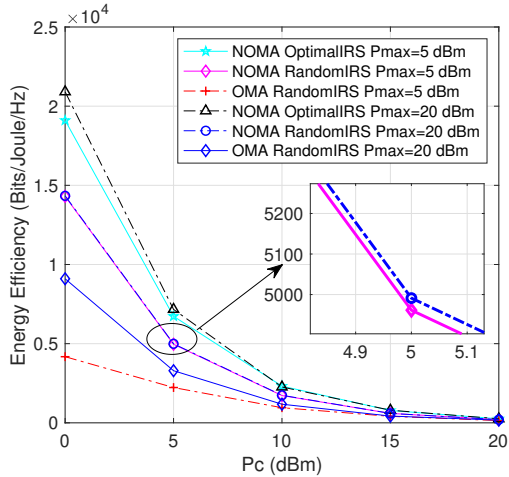}\\
	\caption{EE versus the circuit power $P_c$.} \label{Fig2}
\end{figure}

 Fig. \ref{Fig1} shows the energy efficiency performance versus the number of reflecting elements on the IRS. We set $P_{\max}=10$ dBm and $P_c=10$ dBm. We compare the proposed algorithm with a random phase IRS scheme by considering different numbers of transmit antennas. It can be observed that the energy efficiency of the proposed IRS-NOMA scheme and NOMA with a random phase scheme IRS increases with the number of reflecting elements. Moreover, the proposed scheme always achieves higher energy efficiency than the scheme with a random phase. The performance keeps improving, but the improvement becomes smaller when $M$ increases. This is because the feasible domain of each channel between antennas is narrowed down by increasing $M$ due to the fixed $P_{\max}$. 

Fig. \ref{Fig2} shows the energy efficiency performance versus the circuit power $P_c$ at the BS. We set  $M=20$ and $N=20$. The proposed scheme with different transmit power budgets is considered. It can be shown that the energy efficiency decreases when the circuit power increases. According to the definition of energy efficiency, its value will become smaller when $P_c$ increases. However, the slope of the decreasing curve gets smaller with $P_c$. This is because that the optimization dominates in increasing system energy efficiency when $P_c$ is small. It can also be observed that the MISO IRS-NOMA network with the proposed scheme always outperforms the system with a random phase scheme.  %However, the optimization will not affect the EE performance less than $P_c$ when $P_c$ becomes larger. 

\section{Conclusion}
In this paper, we proposed an effective approach to improve the energy efficiency for the IRS-NOMA network. Specifically, beamforming and phase shift were alternately optimized to achieve the maximum system energy efficiency. Given the fixed phase shift, the beamforming was optimized by introducing some auxiliary variables and applying SCA. Moreover, we used a lower bound and SDR to optimize the phase shift. Compared with the OMA system and random phase scheme, the proposed algorithm can achieve higher energy efficiency. The proposed scheme can also be extended to the multi-user case but with a different phase optimization scheme.

\bibliographystyle{IEEEtran}
%\bibliography{IEEEfull,EEref}
% Generated by IEEEtran.bst, version: 1.14 (2015/08/26)

\end{document}